\documentclass{article}

\usepackage[margin=1in]{geometry}

\usepackage[margin=1in]{geometry}
\usepackage{graphicx}
\usepackage{amsmath}
\usepackage{amssymb}
\usepackage{amsthm}
\usepackage{xcolor}
\usepackage{multirow}
\usepackage{booktabs}
\usepackage[caption=false,font=footnotesize]{subfig}
\usepackage{enumitem}
\usepackage{authblk}

\newtheorem{theorem}{Theorem}
\newtheorem{corollary}{Corollary}

\def\R{\ensuremath{\mathbb{R}}}
\def\H{\ensuremath{\mathbb{H}}}
\newcommand{\dist}[2]{\ensuremath{\textrm{dist}\left(#1, #2\right)}}
\def\Vinv{\ensuremath{V^{-1}}}

\graphicspath{{./Figures/}}

\begin{document}

\title{Transversally exponentially stable Euclidean space extension technique for discrete time systems}
\date{}
\author{Soham Shanbhag\thanks{sshanbhag@kaist.ac.kr} }
\author{Dong Eui Chang\thanks{Corresponding author, dechang@kaist.ac.kr}}
\affil{School of Electrical Engineering, Korea Advanced Institute of Science and Technology, Daejeon, Republic of Korea}

\maketitle
\begin{abstract}
    We propose a modification technique for discrete time systems for exponentially fast convergence to compact sets.
    The extension technique allows us to use tools defined on Euclidean spaces to systems evolving on manifolds by modifying the dynamics of the system such that the manifold is an attractor set.
    We show the stability properties of this technique using the simulation of the rigid body rotation system on the unit sphere $S^3$.
    We also show the improvement afforded due to this technique on a Luenberger like observer designed for the rigid body rotation system on $S^3$.

\end{abstract}

\section{Introduction}


With the proliferation of miniaturized microcontrollers which are discrete in nature, the study of discrete time systems in control theory is gaining importance.
An example of such a system is the rate-integrating gyroscope measuring a change in rotation $\theta_k$ between two instants of time $t = t_k$ to $t = t_{k+1}$, leading to rigid body kinematics of the form
\begin{align*}
    q(k+1) = q(k) \exp(\theta_k/2).
\end{align*}
Moreover, almost all the theory of system identification is in discrete time due to availability of inputs and outputs at every time step while information about the derivatives is rarely available.
While most of the tools for studying control systems are developed for continuous time systems, some tools for linear discrete time systems have been developed, such as the Luenberger observer\cite{Luenb1971} and the Kalman filter\cite{Kalma1960}.
Frequency domain analysis of discrete time systems is also a highly studied topic for linear discrete time systems.

However, such techniques do not extend to nonlinear discrete time systems.
Studies of nonlinear Fourier transforms are sparse, and tools like the Kalman filter are applied to nonlinear systems\cite{KalmaB1961} by linearizing the nonlinear system, which leads to divergence of the estimate from the state in case of highly nonlinear systems.
Other filters like the unscented Kalman filter\cite{JulieU2004} and the particle filter are computationally expensive.
However, some results on feedback linearization of discrete time systems have been addressed\cite{Kotta1995,FliegKN1996} and an observer has been designed for nonlinear discrete time systems\cite{MoraaG1995}.

Moreover, many such systems evolve on manifolds.
The study of such systems evolving on manifolds requires study of properties specific to the manifold, like a notion of distance on the manifold.
An alternative method for such systems is considering the manifolds as embedded in a higher dimensional Euclidean space, and considering the systems as evolving on those Euclidean spaces.
Such techniques have been used for various continuous time applications, such as designing observers for the rigid body rotation and translation system\cite{ShanbC2022a}.
These require techniques designed specifically for such nonlinear systems with the model taken into consideration.
Moreover, using techniques designed for linear discrete time systems for nonlinear discrete time systems evolving on manifolds embedded in Euclidean spaces may lead to the state of the system deviating from the manifold.

One of the methods for tackling this problem for continuous time systems is the was proposed in Feedback Integrator\cite{ChangJP2016}, where the author designs an integration scheme which allows discretization methods to appropriately estimate the state of the system by adding a transversally stable term to the system equation.
This technique has been further used for various applications in continuous time systems, like control system design and Kalman filters \cite{BuC2022,ParkPKC2021,ChangPC2020,Chang2018,K.c.S2023}.
The extension technique has also been extended to handle non-holonomic systems \cite{ChangP2019,ChangPV2022}.

In this paper, we extend the above technique to discrete time systems.
This enables use of linear discrete time techniques for nonlinear systems while respecting the underlying manifold.
This problem is different from that used in applications of \cite{ChangJP2016} where discretization has been performed, since here we study a discrete time system as compared to the discretization of a continuous time system.

\section{Main Results}

Denote by $\R$ the set of real numbers.
The set $f^{-1}(a)$ is defined as $\{x \in dom(f) \mid f(x) = a\}$, where $dom(f)$ is the domain of the map $f$.
The set $f^{-1}(U)$ is then defined as $\cup_{x \in U} f^{-1}(x)$.
The norm of $A \in \R^{n \times m}$ is denoted by $\| A \| = \sqrt{A^T A}$.
Denote the ball of radius $r$ around $x \in \R^n$ as $B_r(x) = \{ y \in \R^n \mid \| y - x \| < r\}$.
Let $\dist{A}{B}$ denote the distance between the sets $A$ and $B$ defined as $\dist{A}{B} = \inf \{ \| x - y \| \mid x \in A, y \in B\}$.
Denote by $\lambda_{max}(A)$ the largest eigenvalue of $A$.

We now present a modification scheme for discrete time systems which converts compact sets to local attractors.
The following theorem provides the main result of this paper.

\begin{theorem}\label{theo:main} 
    Consider a discrete time dynamical system on the subset $U \subset \R^n$
    \begin{align*}
        x_{k+1} = f(x_k, w_k), \quad x_k \in U, \; w_k \in W
    \end{align*}
    where $f$ is a $C^1$ map from $U \times W$ to $U$, where $W \subset \R^m$.
    We make the following assumptions:
    \begin{enumerate}[label=A\arabic*, leftmargin=3em]
        \item \label{assum:main_dyn_inv} There is a $C^2$ function $V:U \to \R_{\geq 0}$ such that $\Vinv(0)$ is compact and nonempty, and $V \circ f(x, w) = V(x)$ for all $x \in U$ and $w \in W$.
        \item \label{assum:main_converge} There exist $b, \epsilon > 0$ such that $\Vinv([0, \epsilon])$ is compact and $\left\| \nabla V(x) \right\|^2 \geq b V(x)$ for all $x \in \Vinv([0, \epsilon])$.
        \item \label{assum:main_Sdelta} There is a $\delta > 0$ such that the set $S_\delta = \{ x \in \R^n \mid \dist{\{x\}}{\Vinv([0,\epsilon])} < \delta \} \subset U$.
    \end{enumerate}
    Then, the modified dynamical system on $\Vinv([0, \epsilon])$,
    \begin{align}
        \tilde{x}_{k+1} = f(\tilde{x}_k, w_k) - \alpha\nabla V(f(\tilde{x}_k, w_k)),\label{eq:cor_extension}
    \end{align}
    converges to the set $\Vinv(0)$ exponentially fast if $\alpha$ is chosen such that $0 < \alpha < \min\left(\frac2d, \frac\delta{L}\right)$, where $L$ is the maximum value of $\| \nabla V(x) \|$ for all $x \in \Vinv([0, \epsilon])$, and $d > 0$ is chosen such that $d > \lambda_{max}(D^2 V(x))$ for all $x \in S_\delta$.
\end{theorem}

\begin{proof}
    We first show that for all $\tilde{x}_k \in \Vinv([0, \epsilon])$, the state $\tilde{x}_{k+1} \in S_\delta$.
    Since $\Vinv([0, \epsilon])$ is compact and $V$ is a $C^2$ function, $\| \nabla V(x) \| < L$ for all $x \in \Vinv([0, \epsilon])$.
    Note that for all $\tilde{x}_k \in \Vinv([0, \epsilon])$ and $w_k \in W$, we have from Assumption~\ref{assum:main_dyn_inv} that $f(\tilde{x}_k, w_k) \in \Vinv([0, \epsilon])$.
    Hence, $\tilde{x}_{k+1} \in B_{\alpha L}(f(\tilde{x}_k, w_k)) \subset S_\delta$ since $\alpha L < \delta$.

    From equation \eqref{eq:cor_extension}, we have that
    \begin{align*}
        V(\tilde{x}_{k+1}) = V(f(\tilde{x}_k, w_k) - \alpha\nabla V(f(\tilde{x}_k, w_k))).
    \end{align*}
    For $\tilde{x}_k \in \Vinv([0, \epsilon])$, using the Lagrange form of the remainder of Taylor series expansion, and noting that  due to continuity of $V$, there exists a $c_k$ between the line joining $f(\tilde{x}_k, w_k)$ and $\tilde{x}_{k+1}$ due to Assumption~\ref{assum:main_Sdelta} such that
    \begin{align*}
        V(\tilde{x}_{k+1}) & = V(f(\tilde{x}_k, w_k)) - \alpha \big\| \nabla V(f(\tilde{x}_k, w_k)) \big\|^2 + \frac{\alpha^2}{2} \nabla V^T(f(\tilde{x}_k, w_k)) D^2 V(c_k) \nabla V(f(\tilde{x}_k, w_k))\\
        & = V(\tilde{x}_k) - \alpha \big\| \nabla V(f(\tilde{x}_k, w_k)) \big\|^2 + \frac{\alpha^2}{2} \nabla V^T(f(\tilde{x}_k, w_k)) D^2 V(c_k) \nabla V(f(\tilde{x}_k, w_k))\\
        &\leq V(\tilde{x}_k) -\alpha \big\| \nabla V(f(\tilde{x}_k, w_k)) \big\|^2 + \frac{\alpha^2}{2}\lambda_{max}(D^2 V(c_k)) \big\| \nabla V(f(\tilde{x}_k, w_k)) \big\|^2
    \end{align*}
    where we have used the property that for any $A \in \R^{n \times n}$, $x^T A x \leq \lambda_{max}(A) \| x \|^2$.
    Since $\tilde{x}_{k+1} \in B_{\alpha L}(f(\tilde{x}_k, w_k))$ implies that $c_k \in B_{\alpha L}(f(\tilde{x}_k, w_k)) \subset S_\delta$, and from the choice of $d$ we have
    \begin{align*}
        V(\tilde{x}_{k+1}) &\leq V(\tilde{x}_k) - \left(\alpha - \frac{\alpha^2}{2} d \right) \big\| \nabla V(f(\tilde{x}_k, w_k)) \big\|^2.
    \end{align*}
    Since $0 < \alpha < 2/d$, the value of $V(x_k)$ is non-increasing in $k$.
    Substituting $c := \alpha - \alpha^2 d/2 > 0$, we have
    \begin{align*}
        V(\tilde{x}_{k+1}) &\leq V(\tilde{x}_k) - c \big\| \nabla V(f(\tilde{x}_k, w_k)) \big\|^2.
    \end{align*}
    Since $\tilde{x}_k \in V^{-1}([0, \epsilon])$ implies $\tilde{x}_{k+1} \in S_\delta$, we have $V(\tilde{x}_{k+1}) \geq 0$.
    Hence, $\| V(f(\tilde{x}_k, w_k)) \|^2 \leq \frac1c V(\tilde{x}_k) = \frac1c V(f(\tilde{x}_k, w_k))$.
    From Assumption~\ref{assum:main_converge}, $\| \nabla V(f(\tilde{x}_k, w_k)) \|^2 \geq b V(f(\tilde{x}_k, w_k))$, hence $b \leq \frac1c$, or $1 - b c \geq 0$.
    Then
    \begin{align*}
        V(\tilde{x}_{k+1}) \leq V(\tilde{x}_k) - c \| \nabla V(f(\tilde{x}_k, w_k)) \|^2 \leq V(\tilde{x}_k) - bc V(f(\tilde{x}_k, w_k)) = V(\tilde{x}_k) - bc V(\tilde{x}_k) = (1 - bc) V(\tilde{x}_k).
    \end{align*}
    Since $|1-bc| < 1$, $V(\tilde{x}_k) \to 0$ as $t \to \infty$ exponentially fast. Hence, $V^{-1}(0)$ is exponentially stable with respect to the dynamics in system~\eqref{eq:cor_extension}.
\end{proof}

This theorem shows that by a suitable choice of $\alpha$, a discrete time system can be modified such that a compact set defined as a zero set of a $C^2$ function is a local attractor of the system.
Moreover, the convergence rate is exponential.
The above theorem is also applicable when the system is autonomous, with a similar proof.
The following corollary presents the conditions for the same.

\begin{corollary}\label{cor:main} 
    Consider a discrete time dynamical system on the subset $U \subset \R^n$
    \begin{align*}
        x_{k+1} = f(x_k), \quad x_k \in U,
    \end{align*}
    where $f$ is a $C^1$ map from $U$ to itself. We make the following assumptions:
    \begin{enumerate}[label=CA\arabic*]
        \item \label{assum:cor_dyn_inv} There is a $C^2$ function $V:U \to \R_{\geq 0}$ such that $\Vinv(0)$ is compact and nonempty, and $V \circ f = V$ on $U$.
        \item \label{assum:cor_converge} There exist $b, \epsilon > 0$ such that $\Vinv([0, \epsilon])$ is compact and $\left\| \nabla V(x) \right\|^2 \geq b V(x)$ for all $x \in \Vinv([0, \epsilon])$.
        \item \label{assum:cor_Sdelta} There is a $\delta > 0$ such that the set $S_\delta = \{ x \in \R^n \mid \dist{\{x\}}{\Vinv([0,\epsilon])} < \delta \} \subset U$.
    \end{enumerate}
    Then, the modified dynamical system on $\Vinv([0, \epsilon])$,
    \begin{align}
        \tilde{x}_{k+1} = f(\tilde{x}_k) - \alpha\nabla V(f(\tilde{x}_k)), \label{eq:main_extension}
    \end{align}
    converges to the set $\Vinv(0)$ exponentially fast for $\alpha$ chosen as specified in Theorem~\ref{theo:main}.
\end{corollary}

\begin{proof}
    The proof follows from that of Theorem \ref{theo:main} and is omitted here.
\end{proof}

The above theorem can be used to extend systems evolving on compact manifolds to ambient Euclidean spaces using manifold embedding theorems.
This also shows the permissiveness of the assumptions on $V$. 
Since many systems, like the rigid body system used in quadcopter and satellite modelling, evolve on the special orthogonal group, the proposed theorem can be used for a huge class of problems.
Moreover, we here show a simple way to derive the function $V$ from the manifold.
The specified choice of $V$ and $U$ given in the proof of the theorem enable us to choose $\epsilon$ and $\alpha$ such that the assumptions of Theorem~\ref{theo:main} are satisfied.

\begin{theorem}\label{theo:manifold}
    Consider a system evolving on the compact manifold $M := g^{-1}(0)$ given by the difference equation
    \begin{align*}
        x_{k+1} = f(x_k, w_k), \quad x_k \in M
    \end{align*}
    where $f$ is a $C^1$ map from $M$ to itself, and $g: \R^n \to \R^k$ is a $C^2$ function.
    Consider the modified dynamical system
    \begin{align*}
        \tilde{x}_{k+1} = f(\tilde{x}_k) - \alpha D g(f(\tilde{x}_k))^T g(\tilde{x}_k) , \quad \tilde{x}_k \in g^{-1}(B_\epsilon(0)),
    \end{align*}
    where $\epsilon$ is such that $g^{-1}(B_\epsilon(0))$ is compact.
    The modified dynamical system converges to the manifold $M$ if:
    \begin{enumerate}[label=A\arabic*', leftmargin=2.5em]
        \item \label{assum:manifold_dyn_inv} For some $\varepsilon > \epsilon$ such that $g^{-1}(B_\varepsilon(0))$ is connected, $g(f(x)) = g(x)$ for all $x \in g^{-1}(B_\varepsilon(0))$.
        \item \label{assum:manifold_converge} For all $x \in g^{-1}(B_\epsilon(0))$, there exists a $b > 0$ such that $\sigma_{min}(D g(x)) \geq b$.
        \item \label{assum:manifold_bounded} For all $x \in g^{-1}(B_\varepsilon(0))$, $\sigma_{max} (D g(x))$ and ${\lambda_{max}(D^2 g_i(x)), \; i = 1, \ldots, k}$ are bounded, where $g_i(x)$ is the $i^{th}$ component of $g(x)$.
    \end{enumerate}
    Here, $\alpha$ is chosen such that
    \begin{align*}
    0 < \alpha < \min\left(\frac2d, \frac{2\delta}{L \epsilon^2}\right),
    \end{align*}
    where $L = \sigma_{max}(Dg(x))$ for all $x \in g^{-1}(B_\epsilon(0))$, $\delta > 0$ such that $\delta < \dist{g^{-1}(bd(B_\varepsilon(0)))}{g^{-1}(B_\epsilon(0))}$ and $d > 0$ is such that $d > \lambda_{max}(D g(x)^T D g(x) + \sum_i g_i(x) D^2 g_i(x))$ for all $x \in g^{-1}(B_\varepsilon(0))$.
\end{theorem}

\begin{proof} 
    Choose
    \begin{align*}
        V(x) = \frac12 \| g(x) \|^2  \textrm{ and } U = g^{-1}(B_\varepsilon(0))
    \end{align*}
    such that
    \begin{align*}
        \nabla V(x) = D g(x)^T g(x) , \textrm{ and } D^2 V(x) = D g(x)^T D g(x) + \sum_i g_i(x) D^2 g_i(x).
    \end{align*}
    Then, our choice of $V$ along with Assumption~\ref{assum:manifold_dyn_inv} satisfies Assumption~\ref{assum:main_dyn_inv}.
    To see that Assumption~\ref{assum:main_converge} is satisfied, note that
    \begin{align*}
        \| D g(x)^T g(x) \|^2 \geq \sigma_{min}(D g(x))^2 g(x)^T g(x) \geq b^2 g^T(x) g(x)
    \end{align*}
    from Assumption~\ref{assum:manifold_converge}.

    We now show that our system satisfies Assumption~\ref{assum:main_Sdelta} of Theorem~\ref{theo:main}, i.e. $S_\delta \subset U$.
    For this system, the set $S_\delta$ is defined as $\{ x \in \R^n \mid \dist{\{x\}}{g^{-1}(B_\epsilon(0))} < \delta \}$.
    From the definition of $S_\delta$, if $x \in S_\delta$, then $x \in B_\delta(y)$ for some $y \in g^{-1}(B_\epsilon(0))$.
    Note that $V(y) \leq \frac12 \epsilon^2$.
    Let us assume that $x \notin U$.
    Then, $V(x) > \frac12 \varepsilon^2 > \frac12 \epsilon^2$.
    By the intermediate value theorem, from the continuity of $g$ and $V$, there exists a $z$ on the line joining $x$ and $y$ such that $V(z) = \frac12 \varepsilon^2$.
    This implies $z \in B_\delta(y) \subset S_\delta$.
    However, $\delta < \dist{g^{-1}(bd(B_\varepsilon(0)))}{g^{-1}(B_\epsilon(0))}$.
    This is a contradiction.
    Hence, $x \in U$, i.e. for any $x \in S_\delta$, $x \in U$.
    This shows that our system satisfies Assumption~\ref{assum:main_Sdelta}.

    In Theorem~\ref{theo:main}, the closure of the set $S_\delta$ is compact, hence the second derivative of $V$ is bounded.
    However, since here we do not assume the boundedness of the set $g^{-1}(B_\varepsilon(0))$, we must assume Assumption~\ref{assum:manifold_bounded} for existence of $d$ since
    \begin{align*}
        \lambda_{max} (D g(x)^T D g(x) + \sum_i g_i(x) D^2 g_i(x)) &\leq \lambda_{max} (D g(x)^T D g(x)) + \sum_i  g_i(x) \lambda_{max}(D^2 g_i(x))\\
        &= \sigma_{max} (D g(x)) + \sum_i g_i(x) \lambda_{max}(D^2 g_i(x)),
    \end{align*}
    which is bounded due to Assumption~\ref{assum:manifold_bounded} and boundedness of $g(x)$.
    To calculate the bounds of $\alpha$, note that
    \begin{align*}
        \frac12 L \epsilon^2 &\geq \sigma_{max}(Dg(x)) \| g(x) \|^2 \geq \| Dg(x)^T g(x) \|
    \end{align*}
    for all $x \in g^{-1}(B_\epsilon(0))$.
\end{proof} 

The manifold extension formulation is beneficial when tools developed for Euclidean spaces are to be used for systems on manifolds.
Usually, study of systems evolving on manifolds require the notion of distance defined on the manifold, and an ``error term'' to be suitably defined using this distance.
An alternative method would be embedding the manifold into ambient Euclidean spaces, and using the subtraction operator defined on the Euclidean space.
However, when using this method, the manifold may not be forward invariant, i.e. the state of the system may diverge from the manifold.
A transversally exponentially stable term will lead to the system state to converge to the manifold, and thus lead to improved results.
An example of this is given in subsection~\ref{subsec:sim_obs_h1}, where we design a Luenberger observer for a system evolving on $S^3$ and compare two observers, one with the extension term and one without the extension term.
We consider randomly chosen initial conditions and show that in all cases, the observer with the transversally stable extension term estimates the state of the system much better as compared to the one without the extension term.

The manifold formulation also shows that although the restrictions on the choice of $V$ exist, they are permissive enough that for many systems, a suitable $V$ can be chosen.
Moreover, the choice of $V$ may be straightforward in a lot of cases, as shown by the choice in Theorem~\ref{theo:manifold}.

\section{Simulation}
We now present simulations for showing the effectiveness of this technique.
To show convergence properties, we perform Monte Carlo simulations.
The simulation is that of a rigid body rotation system where the state is defined using quaternions.
This system evolves on $S^3$, the unit sphere in the quaternion space.
We show the convergence of the system state of the quaternion to $S^3$.
Lastly, to show implementation in real world systems, we show the effect the addition of the extension term leads to in the case of a Luenberger like observer.

\subsection{Simulation on the unit sphere in quaternion space}\label{subsec:sim_H1}
The algebra of quaternions is denoted by $\H$, and let $S^3$ denotes the set of unit quaternions in $\H$.
A vector quaternion is a quaternion whose real part is $0$.
The quaternion $[1, 0, 0, 0]$ is denoted by $1$.
The conjugate quaternion of $q$ is denoted by $\bar{q}$.
The norm of the quaternion $q$ is denoted by $\| q \| = \sqrt{\bar{q} q}$.

We consider the system equations
\begin{align}\label{sys:sys_H1}
    q_{k+1} = q_k \exp(\Omega_k/2),
\end{align}
where $q_k \in S^3$ denotes the rotational state of the body, and $\Omega_k \in \H$ is a vector quaternion denoting the angular rotation of the body.

Define
\begin{align*}
    V(x) = \| \bar{x} x - 1 \|^2, \quad x \in \H, \textrm{ and } U = \H.
\end{align*}
Using the above choice of $V$, we get that 
\begin{align*}
    \nabla V(x) = 4 x (\bar{x} x - 1) \textrm { and } \lambda_{max}(D^2 V(x)) = 12 x^T x.
\end{align*}
Note that the assumptions of Theorem~\ref{theo:main} hold here, since:
\begin{itemize}
    \item $V(f(x_k, \Omega_k)) = \| \exp(-\Omega_k/2) \bar{x} x \exp(\Omega_k/2) - 1 \|^2 = \| \bar{x} x - 1 \|^2 = V(x)$, and $V^{-1}(0) = \{ x \mid \bar{x} x = 1 \}$ is compact and nonempty.
    \item For any $\epsilon < 1$, $V^{-1}([0, \epsilon]) = \{ x \mid 1 - \sqrt{\epsilon} \leq \bar{x} x \leq 1 + \sqrt{\epsilon} \}$ is compact, and choosing $b = 16 (1 - \sqrt{\epsilon})$,
        \begin{align*}
            \| \nabla V(x) \|^2 = 16 \| x \|^2 V(x) \geq b V(x).
        \end{align*}
\end{itemize}
Then, the system equations are modified using the proposed extension technique as
\begin{align}\label{sys:sys_H1_ext}
    \tilde{q}_{k+1} = \tilde{q}_k \exp(\Omega_k/2) - 4 \alpha \tilde{q}_k (\bar{\tilde{q}}_k \tilde{q}_k - 1) \exp(\Omega_k/2)
\end{align}
where $\tilde{q}_k \in \H$ is the state of the quaternion.

For the simulation, consider $\epsilon = 0.5$, $\delta = 0.059$, such that $L = 3.69$ and $d = 15.87$.
Then $\alpha < 0.016$.
Specifically, we choose $\alpha = 0.01$.
The angular rotation of the system is chosen uniformly between $(0, 0, 0, 0)$ and $(0, 0.5, 0.5, 0.5)$.
For showing the convergence of the state of the system to the manifold $S^3$, we perform simulation with the initial condition given by $[w, 0, 0, 0]$, where $w$ is a random number chosen uniformly between $\sqrt{1 - \sqrt{\epsilon}}$ and $\sqrt{1 + \sqrt{\epsilon}}$.
We choose 1000 such initial conditions and run each simulation for 100 epochs.

\begin{figure}[!t]
    \centering
    \resizebox{0.6\linewidth}{!}{\input{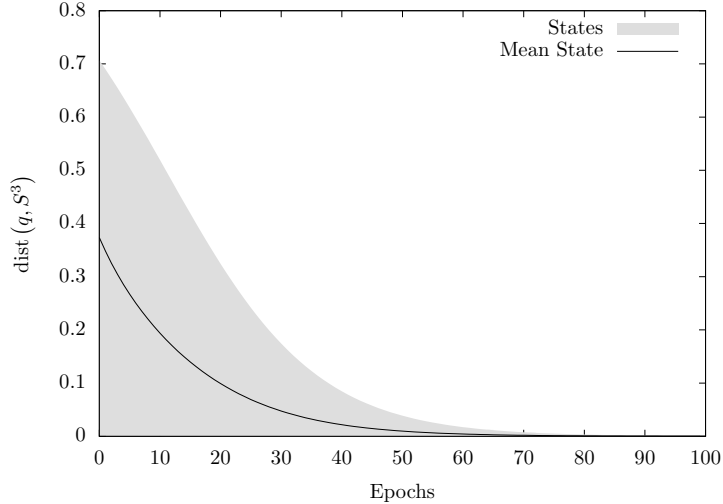}}
    \caption{Convergence of the state of the quaternion system to the manifold $S^3$}
    \label{fig:quat_conv_h1}
\end{figure}

\begin{figure}[t]
    \centering
    \hfill
    \subfloat[][Convergence of the mean of states of observers to the system state]{\resizebox{0.46\linewidth}{!}{\input{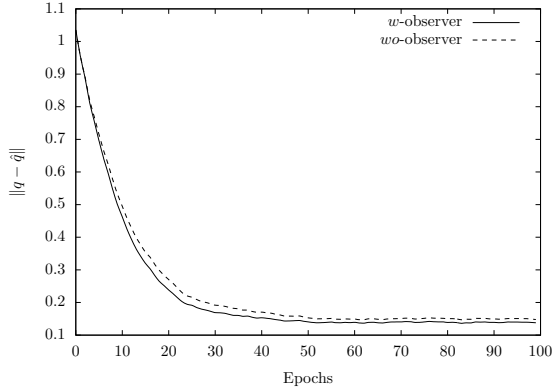}}}\hfill
    \subfloat[][Convergence of observer state with maximum error from the system state to the system state]{\resizebox{0.46\linewidth}{!}{\input{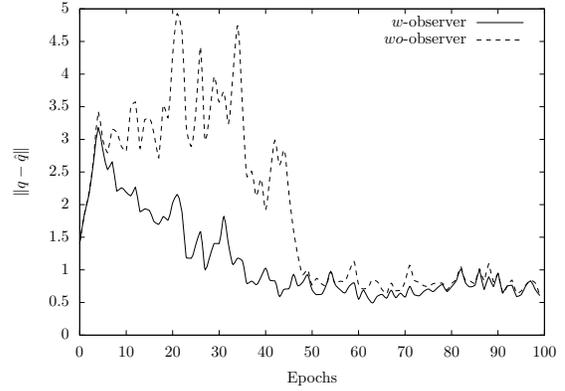}}}\hfill\\\hfill
    \subfloat[][Convergence of the mean of the states to the manifold]{\resizebox{0.46\linewidth}{!}{\input{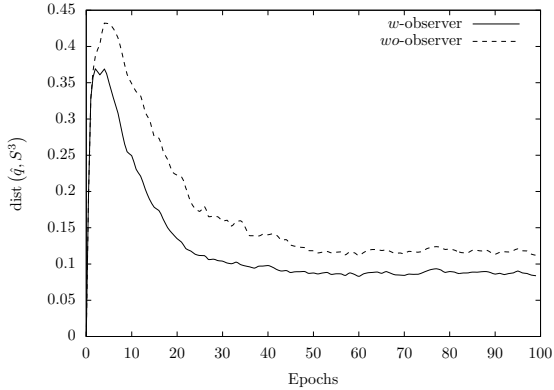}}}\hfill
    \subfloat[][Convergence of the observer state farthest from the manifold to the manifold]{\resizebox{0.46\linewidth}{!}{\input{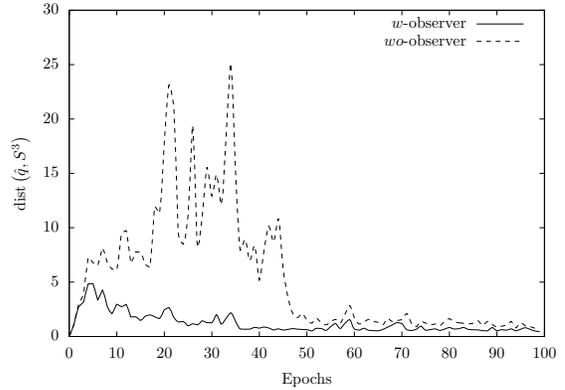}}}\hfill
    \caption{Convergence of the state of the quaternion system to the manifold $S^3$}
    \label{fig:quat_lu}
\end{figure}

The distance of the state from the manifold $S^3$ can be seen in Figure~\ref{fig:quat_conv_h1}.
The black line shows the mean of all the states of the system.
It can be seen from the simulation that the mean of the states of the system converges to the underlying manifold exponentially fast.
This can be seen for all the states starting in $\Vinv([0, \epsilon])$.
The gray shaded region shows the spread of the states, which is highest at the initial epoch due to the randomly chosen initial states.
The spread also exponentially descreases, implying that all the systems eventually converge to the manifold.

\subsection{Simulation of observer on $S^3$}\label{subsec:sim_obs_h1}
We show the effect of using such an additive term in designing observers for systems.
A formal proof and the study of the convergence properties along with comparisons would be a complete work on its own, hence we only show Monte Carlo simulations here.
We consider the system defined in subsection~\ref{subsec:sim_H1}.

We first note that to design an observer for this system on the manifold, a notion of distance on the manifold between two elements of the manifold is to be defined.
An observer then needs to be designed such that this distance tends to $0$ as $t \to \infty$.
We sidestep this requirement by considering the manifold to be embedded in the ambient Euclidean space, and using the Euclidean distance as the distance between two point on the manifold.
However, this formulation will push the state of the system outside the manifold, whereas a desirable property is to restrict the state to the manifold.
To remedy this, the system~\eqref{sys:sys_H1} is extended to the ambient Euclidean space as shown in equation~\eqref{sys:sys_H1_ext}.
To show the improvement due to the extension term, the performance of an observer designed using the modified system~\eqref{sys:sys_H1_ext} is compared with an observer designed using the system \eqref{sys:sys_H1} embedded into $\H$ without the extension term.

Motivated by a Luenberger like observer system, the observer equations are considered as
\begin{align*}
    \hat{q}_{w, k+1} &= \hat{q}_{w, k} \exp(\Omega_k/2) + (q_k - \hat{q}_{w, k}) L - 4 \alpha \hat{q}_{w, k} (\bar{\hat{q}}_{w, k} \hat{q}_{w, k} - 1) \exp(\Omega_k/2),\\
    \hat{q}_{wo, k+1} &= \hat{q}_{wo, k} \exp(\Omega_k/2) + (q_k - \hat{q}_{wo, k}) L,
\end{align*}
where the subscript $w$ denotes the observer with the transversally stable term and the subscript $wo$ denotes the observer without this term.
Here $q_k$ denotes the measured state of the system evolving using the difference equation~\eqref{sys:sys_H1}.
For the sake of brevity, we use the term $w$-observer and $wo$-observer for systems containing the transversally stable term and for those without the transversally stable term.

The initial condition for the system state is chosen as a random state in $S^3$.
The initial condition for the observer states are chosen as $1$.
We run 1000 such simulations with each simulation run for 100 epochs.
The state of the system at epoch $k$ is measured as the true state of the system multiplied by $\exp(\nu_k)$ where $\nu_k$ is Gaussian noise with mean $0$ and standard deviation $0.1$.
The angular velocity of the system is chosen as $(0, \omega_x, \omega_y, \omega_z)$ such that $0 < \omega_x, \omega_y, \omega_z < 1$ with the time step between two epochs as $10$, such that $\Omega_k$ is of the form $(0, \theta_x, \theta_y, \theta_z)$ such that $0 < \theta_x, \theta_y, \theta_z < 10$.
We choose $L = (0.11, -0.18, -0.18, -0.18)$ as the mean of $\exp(\Omega_k/2)$ since there exists no choice of $L$ such that $\| \exp(\Omega_k/2) - L \| < 1$ for all $\Omega_k$.
The constants $\epsilon, \delta$ and $\alpha$ are chosen as in subsection \ref{subsec:sim_H1}.

The convergence of the observers to the system state is shown in Figure~\ref{fig:quat_lu}.
Here, mean state denotes the mean of all the states of the system.
The maximum error in the observer states, minimum error in the observer states (both calculated for all $k > 25$ and among all simulations) and the final error mean and standard deviation are shown in Table~\ref{tab:errs}.
We note that while both the observers converge to the system state, the mean of the error of the $w$-observer is less than that of the $wo$-observer.
This difference in mean is significant in case of the distance from the $S^3$ manifold.
If the observer state is intended to be restricted to $S^3$, by projecting the unrestricted observer state to $S^3$, the error in the state of the observer from the manifold also contributes to the error in the observer.
Moreover, the maximum error in case of the $w$-observer is also considerably less than that of the $wo$-observer, implying that the variation of the error in case of the $w$-observer is also lower than that of the $wo$-observer.
This difference in variation can be seen in the standard deviation of the observers, where the standard deviation of the distance of the $wo$-observer state from the manifold is 1.5 times that of the $w$-observer.
This simulation shows the benefits of extending the system to the ambient space using a transversally stable term before designing the observer system.

\begin{table*}[t]
    \centering
    \caption{\label{tab:errs} Error in the observer states for Luenberger observer system.}
    \begin{tabular}{*{9}{c}}
        \toprule
        \multirow{2}{*}{Observer} & \multicolumn{4}{c}{$\| q - \hat{q} \|$} & \multicolumn{4}{c}{$\dist{q}{S^3}$} \\
        \cmidrule(lr){2-5} \cmidrule(lr){6-9}
        & Mean & Std & Max & Min & Mean & Std & Max & Min\\
        \midrule
            $w$-observer & 0.138 &  0.07  &  1.816 &  0.007 &  0.084 &  0.072 &  2.182 &  0. \\
            $wo$-observer & 0.148 &  0.076 &  4.747 &  0.007 &  0.112 &  0.101 & 25.121 &  0. \\
        \bottomrule
    \end{tabular}
\end{table*}


\section{Conclusion}
In this paper, we extend the ambient space transversally stable extension technique\cite{ChangJP2016} to discrete time systems.
The use of transversally stable extension has been used to improve existing results in control theory in the case of continuous time systems.
We expect that the provided formulation in Theorem~\ref{theo:main} enables such improvement in the case of discrete time systems.
An example of improvements afforded by such a formulation has been shown in subsection~\ref{subsec:sim_obs_h1}.
A quantitative analysis of such improvement for different problems in control theory and their comparison with corresponding manifold respecting formulations would be a future area of research.

\bibliographystyle{plain}
\bibliography{References.bib}

\end{document}